\documentclass[reqno,12pt]{amsart}

\usepackage{amssymb}
\usepackage{amsthm}
\usepackage{amsmath}
\usepackage{amsfonts}
\usepackage{mathtools}
\usepackage{enumitem}
\usepackage[colorlinks, 
linkcolor=red,
anchorcolor=magenta,
citecolor=blue
]{hyperref}
\usepackage[margin=1.5in]{geometry}

\title[Quantum dynamics with Liouville discrepancies]{Quantum Dynamical Bounds for Quasi-Periodic  Operators with Liouville Frequencies}

\date{\today}

\author[]{Matthew Bradshaw}
\address[Matthew Bradshaw]{
	University of Connecticut, Storrs, CT, 06269, USA.
    }
\email{matthew.bradshaw@uconn.edu}

\author[]{Titus de Jong}
\address[Titus de Jong]{
	University of California Irvine, Irvine, CA, 92697, USA.
}
\email{typdejong@gmail.com, tydejong@uci.edu}

\author[]{Wencai Liu}
\address[Wencai Liu]{
	Department of Mathematics, Texas A\&M University, College Station, TX, 77843, USA.
}
\email{liuwencai1226@gmail.com, wencail@tamu.edu}

\author[]{Audrey Wang}
\address[Audrey Wang]{
	University of Chicago, 5801 S Ellis Ave, Chicago, IL 60637, USA.
}
\email{audrey.yx.wang@gmail.com}

\author[]{Xueyin Wang}
\address[Xueyin Wang]{
	Department of Mathematics, Texas A\&M University, College Station, TX, 77843, USA.
}
\email{xueyin@tamu.edu}

\author[]{Bingheng Yang}
\address[Bingheng Yang]{
	Department of Mathematics, Texas A\&M University, College Station, TX, 77843, USA.
}
\email{bhyang@tamu.edu}

\theoremstyle{plain}
\newtheorem{theorem}{Theorem}[section]
\newtheorem{corollary}[theorem]{Corollary}
\newtheorem{lemma}[theorem]{Lemma}

\theoremstyle{definition}
\newtheorem{definition}[theorem]{Definition}

\newtheorem{remark}[theorem]{Remark}

\DeclareMathOperator{\leb}{Leb}
\DeclareMathOperator{\dist}{dist}

\begin{document}

\begin{abstract}
    We establish quantum dynamical upper bounds for quasi-periodic Schr\"odinger operators with Liouville frequencies. Our approach combines semi-algebraic discrepancy estimates for the Kronecker sequence $\{n\alpha\}$ with quantitative Green's function estimates adapted to the Liouville setting.
\end{abstract}

\maketitle	

\section{Introduction}

In this paper, we study the quantum dynamics of one-dimensional discrete Schr\"{o}dinger operators $H$. It is well-known that the solution of the time-dependent Schr\"{o}dinger equation $i\partial _t\psi=H\psi$ is given by $\psi(t)=e^{-itH} \psi(0)$.
For simplicity, we assume that the initial condition $\psi(0)=\phi$ has compact support.
The position operator $X$ is defined as
\begin{equation*}
    (X\psi)_{n}=n\psi_{n}.
\end{equation*}
To describe the evolution of $\psi(t)$, we focus on the time-averaged $p^{\rm th}$ moment of the position operator defined via
\begin{equation}\label{pm}
    \langle |X_{H}|_{\phi}^{p}\rangle(T):=\frac{2}{T}\int_{0}^{\infty} e^{-2t/T}\langle \psi(t),|X|^{p}\psi(t) \rangle d t.
\end{equation}
Thus the growth of $T\mapsto\langle |X_{H}|_{\phi}^{p}\rangle(T)$ reflects the speed of which the particles spread out.

In particular, we consider the quasi-periodic Schr\"odinger operators $H=H_{\theta}$ on $\ell^2(\mathbb{Z})$,
\begin{equation}\label{H}
    (H_{\theta}\psi)_n=\psi_{n+1}+\psi_{n-1}+V(\theta+n\alpha)\psi_n,
\end{equation}
where $V\in C^{\omega}(\mathbb{T},\mathbb{R})$ is the potential, $\theta\in\mathbb{T}$ is the phase, and $\alpha\in\mathbb{R}\setminus\mathbb{Q}$ is the frequency. The quantum dynamical behavior of $H_{\theta}$ depends heavily on the arithmetic property of the frequency $\alpha$ \cite{MR1423040, MR1663518, MR4404788, MR4578350}. We say $\alpha$ is Diophantine if $\|n\alpha\|_{\mathbb{T}} \geqslant \eta |n|^{-\gamma}$ for some $\eta>0,\gamma\geqslant 1$ where $\|\cdot\|_{\mathbb{T}}:=\dist(\cdot,\mathbb{Z})$. We say $\alpha$ is Liouville if it is not Diophantine.

It has been shown that in the regime of large potential $V$ and Diophantine frequency $\alpha$, the quantity $\langle |X_{H_{\theta}}|_{\phi}^{p}\rangle(T)$ remains finite for almost every $\theta\in\mathbb{T}$, and we achieve the so-called dynamical localization, see \cite{MR2100420,MR4637128,MR1796713,MR4216568}. However, dynamical localization fails to hold for generic $\theta \in \mathbb{T}$ \cite{JitoSimonSC, MR4756946, jitomirskaya2024sharp}. 

For Diophantine frequency $\alpha$, the growth behavior of $\langle |X_{H_{\theta}}|_{\phi}^{p}\rangle(T)$ has been investigated through various methods. 
In \cite{MR4564259}, Liu established quantum dynamical upper bounds for long-range operators on $\mathbb{Z}^{d}(d\geqslant 1)$ based on the sublinear bounds for the semi-algebraic discrepancy:
\begin{equation}\label{sl}
    \# \big\{\, |n| \leqslant N : \theta + n\alpha \bmod\mathbb{Z} \in \Theta \,\big\} \leqslant N^{1 - \delta},
\end{equation}
where $\Theta$ is a semi-algebraic set with suitable complexity. In \cite{prema,LPW}, the authors further developed this approach by applying tools from analytic number theory to analyze the discrepancy of semi-algebraic sets in the setting of quasi-periodic operators  with multi-frequency shift and skew-shift potentials. More recently, Shamis–Sodin \cite{MR4604835} developed a method applicable to operators on $\mathbb{Z}^{d}$, using quantitative Green's function estimates from \cite{MR4546503}. The arguments in \cite{MR4564259, MR4604835} both require that the bad boxes are sufficiently rare so that they cannot fill any line of a large scale box (i.e., a sublinear bound as in \eqref{sl}).

All known results establishing power-logarithmic upper bounds of the form $(\log T)^C$ crucially rely on the Diophantine condition for $\alpha$. This naturally leads to the question: What is the growth behavior of $\langle |X_{H_{\theta}}|_{\phi}^{p}\rangle(T)$ when $\alpha$ is Liouville? For Liouville frequencies, sublinear discrepancy bounds fail, rendering the methods of \cite{MR4564259,prema,LPW} inapplicable.

Our proof is inspired by Jitomirskaya–Liu \cite{MR4288185}, which established a  $T^{\varepsilon}$ bound under the existence of a single suitably chosen box of the Green’s function. A key observation in the present work is that a similar one-box assumption already suffices to obtain a power–logarithmic bound of the form $(\log T)^C$. The argument in \cite{MR4288185} relied on the monotonicity of the transport exponent with respect to $p$, which, however, led to the loss of a sharper estimate for $\langle |X_{H}|_{\phi}^{p}\rangle(T)$. In our approach, we decompose the estimate of $\langle |X_{H}|_{\phi}^{p}\rangle(T)$ across different scales, which allows us to achieve an improved bound.

In this paper, we first prove a general criterion (Theorem~\ref{General Transport Bound}) for deriving quantum dynamical upper bounds. Notably, this criterion applies even to long-range operators without underlying dynamical systems. The criterion is based solely on the existence of a single suitably chosen box of the Green’s function. 

We then turn to discrepancy estimates for shift dynamics $\{n\alpha\}$ on semi-algebraic sets in the case of Liouville frequencies. As a consequence, we verify the existence of a box  of  the desired Green’s function, extending the sublinear bounds known in the Diophantine case to a weaker, yet still effective, setting.

Denote by $L(E)$ the Lyapunov exponent, see Section \ref{prelim} for the definition. As applications to quasi-periodic operators, we obtain the following results:
\begin{theorem}\label{qdDC}
    Let $V\in C^{\omega}(\mathbb{T},\mathbb{R})$ be non-constant and assume $L(E)>\lambda>0$ for all $E\in\mathbb{R}$. Let $\eta>0, \gamma\geqslant 1$. Suppose that $\alpha\in \mathbb{R}$ satisfies
    \begin{equation*}
        \| n\alpha\|_{\mathbb{T}} \geqslant \eta |n|^{-\gamma}\quad \text{for all}\ n\in\mathbb{Z}\setminus\{0\}.
    \end{equation*}
    Then there exists $C_{0}>0$ such that for any $p>0,\varepsilon>0$ and $\phi\in\ell^{2}(\mathbb{Z})$ compactly supported, there exists $T_{1}(\alpha,V,\lambda,\phi,p,\varepsilon)>0$ such that for  $T\geqslant T_{1}$,
    \begin{equation*}
        \sup_{\theta\in\mathbb{T}}\langle |X_{H_{\theta}}|^p_{\phi} \rangle (T) \leqslant (\log T)^{pC_{0}\gamma + \varepsilon}.
    \end{equation*}
\end{theorem}

\begin{theorem}\label{qdWDC}
    Let $V\in C^{\omega}(\mathbb{T},\mathbb{R})$ be non-constant and assume $L(E)>\lambda>0$ for all $E\in\mathbb{R}$. Let $\eta>0, \kappa>0,\gamma > 1$. Suppose that $\alpha \in \mathbb{R}$ satisfies
    \begin{equation*}
        \|n\alpha\|_{\mathbb{T}} \geqslant \eta e^{-\kappa (\log |n|)^{\gamma}} \quad \text{for all}\ n\in\mathbb{Z}\setminus\{0\}.
    \end{equation*}
    Then there exists $C_{0}>0$ such that for any $p>0,\varepsilon>0$ and $\phi\in\ell^{2}(\mathbb{Z})$ compactly supported, there exists $T_{2}(\alpha,V,\lambda,\phi,p,\varepsilon)>0$ such that for  $T\geqslant T_{2}$,
    \begin{equation*}
       \sup_{\theta \in \mathbb{T}} \langle |X_{H_{\theta}}|^p_{\phi} \rangle (T) 
       \leqslant \exp \bigg( p\kappa (C_{0}+\varepsilon)^{\gamma} (\log\log T)^{\gamma}\bigg).
   \end{equation*}
\end{theorem}

\begin{theorem}\label{qdLiou}
    Let $V\in C^{\omega}(\mathbb{T},\mathbb{R})$ be non-constant and assume $L(E)>\lambda>0$ for all $E\in\mathbb{R}$. Let $\eta>0,\kappa>0,0<\gamma<\frac{1}{C_{0}}$. Suppose that $\alpha\in\mathbb{R}$ satisfies
    \begin{equation*}
        \| n\alpha\|_{\mathbb{T}}\geqslant \eta e^{-\kappa |n|^{\gamma}}\quad \text{for all}\ n\in\mathbb{Z}\setminus\{0\}.
    \end{equation*}
    Then for any $p>0,\varepsilon>0$ and $\phi\in\ell^{2}(\mathbb{Z})$ compactly supported, there exists $T_{3}(\alpha,V,\lambda,\phi,p,\varepsilon)>0$ such that for  $T\geqslant T_{3}$,
    \begin{equation*}
        \sup_{\theta\in\mathbb{T}} \langle |X_{H_{\theta}}|_{\phi}^{p}\rangle (T)\leqslant \exp \bigg(p (\log T)^{C_{0}\gamma+\varepsilon}\bigg).
    \end{equation*}
\end{theorem}

\begin{remark}
    In Theorem \ref{qdDC}, \ref{qdWDC}, and \ref{qdLiou}, the constant $C_{0}=5C$, where $C \geqslant 1$ is the constant $C(d)$ from Lemma \ref{covers} when $d=1$.
\end{remark}
Theorem~\ref{qdDC} is not new; it was previously established in \cite{MR4288185, MR4564259, JPo, prema}. In fact, those works provide stronger versions with explicit estimates on the constants, which our approach does not yield. However, we include the result here to demonstrate the flexibility and effectiveness of our method.  Theorems~\ref{qdWDC} and~\ref{qdLiou} are new. In the same setting, we note that Damanik–Tcheremchantsev \cite{MR2291919, MR2464194} proved an upper bound $T^{\varepsilon}$. Theorems~\ref{qdWDC} and~\ref{qdLiou} thus provide quantitative estimates for this $T^{\varepsilon}$ in the case of Liouville frequencies.

The existence of a good box is ensured by the discrepancy estimates for semi-algebraic sets in Section \ref{DiscrepancyEstimates}. To derive the Green's function estimates on the good box, we apply the large deviation theorem for transfer matrices established in \cite{MR4373222}, which remains valid in the Liouville setting.

\section{Preliminaries}
\label{prelim}








\subsection{Transfer matrix and Lyapunov exponent}
Denote by $C^{\omega}_{h}(\mathbb{T},\mathbb{R})$ the space of real-valued bounded analytic functions on the strip $\{\theta : |\Im \theta| < h\}$. For any $V \in C^{\omega}_{h}(\mathbb{T},\mathbb{R})$, define 
\begin{equation*}
    \|V\|_{h} = \sup_{|\Im \theta| < h} |V(\theta)|.
\end{equation*}
Let $C^{\omega}(\mathbb{T},\mathbb{R}) := \bigcup_{h > 0} C^{\omega}_{h}(\mathbb{T},\mathbb{R})$.

Let $E\in\mathbb{R}$. Denote
\begin{equation*}
    S_{E}^{ V}(\theta) := \begin{pmatrix}
        E -  V(\theta) & -1 \\
        1 & 0
    \end{pmatrix}.
\end{equation*}
For any finite interval $\Lambda = [x_{1}, x_{2}] \subseteq \mathbb{Z}$ with $x_{1} < x_{2}$, define the transfer matrix from $x_{1}$ to $x_{2}$ as
\begin{equation*}
    M_{\Lambda}(\theta) := \prod_{k = x_{2} - 1}^{x_{1}} S_{E}^{ V}(\theta + k\alpha).
\end{equation*}
For $H_{\theta}$ defined in \eqref{H}, we let $H_{\Lambda}(\theta):=H_{\Lambda}(V,\theta) := R_{\Lambda} H_{\theta} R_{\Lambda}$, where $R_{\Lambda}$ is the projection onto $\Lambda$. In particular, for $\Lambda = [0, N-1]$, denote $H_{N}(\theta) := H_{[0, N-1]}(\theta)$ and $M_{N}(\theta) := M_{[0, N-1]}(\theta)$.

It is straightforward to verify that
\begin{equation*}
    M_{[x_{1}, x_{2}]}(\theta) = M_{x_{2} - x_{1} + 1}(\theta + x_{1}\alpha).
\end{equation*}
For $N \geqslant 1$, it is well-known (see \cite{MR2100420}) that
\begin{equation}\label{four}
    M_{N}(\theta) = \begin{pmatrix}
        \det(H_{N}(\theta) - E) & -\det(H_{N-1}(\theta + \alpha) - E) \\
        \det(H_{N-1}(\theta) - E) & -\det(H_{N-2}(\theta + \alpha) - E)
    \end{pmatrix},
\end{equation}
with the convention $\det(H_{0}(\theta) - E) := 1$ and $\det(H_{-1}(\theta) - E) := -1$.

Define the Lyapunov exponent as
\begin{equation*}
    L(E)=\lim_{N \to \infty}\frac{1}{N} \int_{\mathbb{T}} \log \| M_{N}(\theta) \| \, d\theta.
\end{equation*}
By the unique ergodicity, one has
\begin{equation}\label{fur}
    L(E)=\lim_{|\Lambda| \to \infty} \sup_{\theta \in \mathbb{T}} \frac{1}{|\Lambda|} \log \|M_{\Lambda}(\theta)\|.
\end{equation}
According to the Thouless formula, the positivity of $L(E)$ for all $E\in\mathbb{R}$ yields that $L(z) > 0$ for all $z \in \mathbb{C}$. 


\subsection{Green's Function}
Denote by $\sigma(H)$ the spectrum of $H$. For $z \notin \sigma(H)$, the Green's function of $H$ at $z$ is defined by
\begin{equation*}
    G(z) := (H - z I)^{-1},
\end{equation*}
and for a finite interval $\Lambda$, define
\begin{equation*}
    G_{\Lambda}(z) := (R_{\Lambda}HR_{\Lambda} - z I)^{-1}.
\end{equation*}
In particular, for $H_{\theta}$ defined in \eqref{H} with $\Lambda = [x_{1}, x_{2}]$ with $x_{1} < x_{2}$ and $x_{1} \leqslant m \leqslant n \leqslant x_{2}$, by Cramer's rule  we have
\begin{equation*}
    G_{\Lambda}(z,\theta)(m,n) = \frac{ \det(H_{m - x_{1}}(\theta + x_{1}\alpha) - z) \cdot \det(H_{x_{2} - n}(\theta + (n+1)\alpha) - z) }{ \det(H_{\Lambda}(\theta) - z) }.
\end{equation*}
Using \eqref{four}, it follows that
\begin{equation}\label{cramer}
    |G_{\Lambda}(z,\theta)(m,n)| \leqslant \frac{ \|M_{[x_{1},m-1]}(\theta)\| \cdot \|M_{[n+1,x_{2}]}(\theta)\| }{ |\det(H_{\Lambda}(\theta) - z)| }.
\end{equation}

\subsection{Semi-algebraic sets}

\begin{definition}[{\cite{MR2100420}}] \label{defsemi}
	We say $\mathcal{S}\subseteq \mathbb{R}^{d}$ is a semi-algebraic set if it is a finite union of sets defined by a finite number of polynomial inequalities. More precisely, let $\{P_{1}, P_{2}, \cdots, P_{s}\}$ be a family of real polynomials to the variables $x=(x_{1}, x_{2}, \cdots, x_{d})$ with $\deg(P_{i})\leqslant  q$ for $i=1,2,\cdots, s$. A (closed) semi-algebraic set $\mathcal{S}$ is given by the expression
	\begin{equation}\label{sas}
		\mathcal{S}=\bigcup_{j} \bigcap_{\ell\in \mathcal{L}_{j}} \{x\in \mathbb{R}^{d}: P_{\ell}(x) \ \varsigma_{j\ell} \ 0\},
	\end{equation}
	where $\mathcal{L}_{j}\subseteq \{1,2,\cdots, s\}$ and $\varsigma_{j\ell}\in \{\geqslant, \leqslant , =\}$. Then we say that the degree of $\mathcal{S}$, denoted by $\deg(\mathcal{S})$, is at most $sq$. In fact, $\deg(\mathcal{S})$ means the smallest $sq$ overall representation as in \eqref{sas}.
\end{definition}


\begin{lemma}[{\cite{MR2100420,MR880035, MR889980,MR3990603}}\label{covers}]
	Let $\mathcal{S}\subseteq [0,1]^{d}$ be a semi-algebraic set of degree $B$. Let $\epsilon>0$ be a small number and $\leb(\mathcal{S})\leqslant  \epsilon^{d}$. Then $\mathcal{S}$ can be covered by a family of $\epsilon$-balls with total number less than $B^{C(d)} \epsilon^{1-d}$.
\end{lemma}

\subsection{Discrepancy}
\begin{definition}[\cite{MR1470456}]\label{defdis}
    Let $\{x_n\}_{n=1}^{N}$ be a sequence in $\mathbb{T}^d$. The discrepancy of $x_{n}$ is defined as
    \begin{equation*}
        D_N(x_n)=\sup_{I\in R} \bigg|\frac{\#\{1\leqslant n\leqslant N: x_{n}\in I\}}{N}-\leb(I)\bigg|,
    \end{equation*}
    where $R$ denotes the family of all axis-aligned rectangles in $\mathbb{T}^d$.
\end{definition}

In particular, for $x_{n}=\theta+n\alpha \bmod \mathbb{Z}^{d}$, we denote by $D_{N}(\alpha)$ the discrepancy for short (which is independent of $\theta$ by Definition \ref{defdis}). 

The Erd\H{o}s–Tur\'{a}n–Koksma inequality provides an upper bound for $D_{N}(x_{n})$.

\begin{theorem}[{\cite[Erd\H{o}s–Tur\'{a}n–Koksma Inequality]{MR1470456}}]\label{ETI}
    Let $\{x_n\}_{n=1}^{N}$ be a sequence in $\mathbb{T}^d$ and $M \in \mathbb{N}$ be arbitrary. Then
    \begin{equation*}
        D_N(x_{n}) \leqslant \bigg(\frac{3}{2}\bigg)^d \bigg(\frac{2}{M+1} + \sum_{0< |m| < M} \frac{1}{r(m)} \bigg|\frac{1}{N} \sum_{n = 1}^N e^{2\pi i \langle m,x_{n}\rangle}\bigg|\bigg),
    \end{equation*}
    where $r(m)=\prod_{i=1}^{d}\max \{1,|m_{i}|\}$ and  $|m|=\max_{1\leqslant  i\leqslant  d} |m_{i}|$.
\end{theorem}

Throughout the paper, we write $A \lesssim_{*} B$ to indicate that there exists a constant $C > 0$ depending on $*$ such that $A \leqslant C B$. 

\section{Criterion for quantum dynamics}\label{CriterionForQuantumDynamics}
In this section, we establish a criterion for quantum dynamical upper bounds based on Green's function estimates on suitable scales. The following criterion works for bounded self-adjoint long-range operators on $\ell^{2}(\mathbb{Z})$.

\begin{theorem}\label{General Transport Bound}
Let $\{V_n\}\in \ell^{\infty}(\mathbb{Z})$ be a real sequence, and consider the long-range operator
\begin{equation*}
    (H\psi)_n = \sum_{m=-\infty}^{\infty} A_m \psi_{n-m} + V_n \psi_n,
\end{equation*}
where $\overline{A_m} = A_{-m}$ and $|A_m| \leqslant C_1 e^{-c_1 |m|}$ for all $m \in \mathbb{Z}$. Assume that $\sigma(H) \subseteq [-K+1, K-1]$ for some $K \geqslant 3$.

Let $\phi \in \ell^2(\mathbb{Z})$ be compactly supported. Suppose that for every $|N| \geqslant N_0$, there exists an interval $I \subseteq [-\frac{|N|}{2}, -\frac{|N|}{4}] \cap \mathbb{Z}$ or $I \subseteq [\frac{|N|}{4}, \frac{|N|}{2}] \cap \mathbb{Z}$ satisfying the following conditions:
\begin{enumerate}
    \item There exists a monotone increasing function $\Psi: \mathbb{R}^+ \rightarrow \mathbb{R}^+$ such that
    \begin{equation}\label{IntervalLargerThanF}
        |I| \geqslant \Psi(|N|) \geqslant (\log |N|)^{C_2},
    \end{equation}
    for some constant $C_2 > 1$.
    
    \item There exists $0<c_2\leqslant c_1$ such that for any $z=E+i\epsilon$ with $|E|\leqslant K$ and $0<\epsilon\leqslant \epsilon_{0}$,
    \begin{equation}\label{GreenExpDecay}
        |G_I(z)(m,n)| < e^{-c_2 |I|} \quad \text{for all } |m - n| > \frac{|I|}{20}.
    \end{equation}
\end{enumerate}

Then for any $p > 0$, there exists $T_0 = T_0(\phi, K, C_1, C_2, c_1, c_2, \epsilon_{0}, p)$ such that for all $T \geqslant T_0$,
\begin{equation*}
    \langle |X_H|^p_\phi \rangle (T) \leqslant \left[ \Gamma\left( \frac{80}{c_2} \log T \right) \right]^p,
\end{equation*}
where $\Gamma: \mathbb{R}^+ \rightarrow \mathbb{R}^+$ is the inverse function of $\Psi$.
\end{theorem}
\begin{proof}
    Let $z=E+\frac{i}{T}$ where $|E|\leqslant K$ and $T\geqslant \frac{1}{\epsilon_{0}}$. Recall the following lemma essentially proved by Jitomirskaya–Liu \cite{MR4288185}:
\begin{lemma}[{\cite[Lemma 2.3]{MR4288185}}]\label{ModifiedLemma2.1}
    Assume for some interval $I$ with $I\subseteq [-\frac{|N|}{2},-\frac{|N|}{4}]$ or $I\subseteq [\frac{|N|}{4}, \frac{|N|}{2}]$,  (\ref{GreenExpDecay}) holds. Then for any $j\in\operatorname{supp} \phi$, 
    \begin{equation*}
        |G(z)(j,N)|\lesssim_{\phi,C_1,c_1} T^{4}e^{-\frac{c_{2}}{20} |I|}.
    \end{equation*}
\end{lemma}

By Lemma \ref{ModifiedLemma2.1} and \eqref{IntervalLargerThanF}, for any $j\in \operatorname{supp} \phi$ and $|n|\geqslant N_{0}$,
\begin{equation}\label{Gz}
    |G(z)(j,n)| \lesssim_{\phi,C_{1},c_{1}} T^{4} e^{-c_{3}\Psi(|n|)},
\end{equation}
where $c_{3}=\frac{c_{2}}{20}$.
Denote
\begin{equation*}
    a(j,n,T)=\frac{2}{T}\int_0^{\infty} e^{-2t/T}|\langle \delta_n, e^{-i t H}\delta_j\rangle|^2 d t.
\end{equation*}
According to Parseval's identity (see \cite{damanik2022one}), one has 
\begin{equation*}
    a(j,n,T)= \frac{1}{\pi T} \int_{\mathbb{R}} |G(z)(j,n)|^{2} dE.
\end{equation*}
Since $\sum_{n\in\mathbb{Z}} |\langle \delta_n, e^{-i t H}\delta_j \rangle|^2=1$, one has
\begin{equation}\label{unit}
    \sum_{n\in\mathbb{Z}} a(j,n,T)=1\quad \text{for any}\ j\in \operatorname{supp} \phi.
\end{equation}
By \eqref{unit}, for any $R\geqslant 1$,
\begin{equation}\label{Xp}
    \begin{split}
        \langle |X_{H}|_{\phi}^p\rangle (T) & \lesssim_{\phi} \bigg( \sum_{|n|\leqslant R}+\sum_{|n|\geqslant R} \bigg) \sum_{j\in \operatorname{supp} \phi} |n|^{p}a(j,n,T)\\
        &\lesssim_{\phi} R^{p}+ \sum_{j\in \operatorname{supp} \phi}\sum_{|n|\geqslant R}|n|^{p}a(j,n,T).
    \end{split}
\end{equation}
By the Combes–Thomas estimate (see \cite{MR1301371}), for sufficiently large $|n|$, 
\begin{equation}\label{CT}
    \begin{split}
        a(j,n,T)&\leqslant  \frac{1}{\pi T} \bigg( \int_{|E|\geqslant K} + \int_{|E|\leqslant K} \bigg)|G(z)(j,n)|^{2} dE\\
        &\lesssim_{K} \frac{1}{T}e^{-c_{4}|n|} + \frac{1}{T}\int_{|E|\leqslant K}  |G(z)(j,n)|^{2} dE,
    \end{split}
\end{equation}
where $c_{4}>0$ is a universal constant. By \eqref{Gz}, we have
\begin{equation}\label{GT}
    \begin{split}
        \frac{1}{T}\int_{|E|\leqslant K}  |G(z)(j,n)|^{2} dE& \lesssim_{\phi,C_{1},c_{1},K} T^{7} e^{-2c_{3}\Psi(|n|)}.
    \end{split}
\end{equation}
Combining \eqref{Xp}, \eqref{CT} and \eqref{GT}, we get
\begin{equation*}
    \begin{split}
        \langle |X_{H}|_{\phi}^p\rangle (T) &\lesssim_{\phi,K,C_{1},c_{1}} R^{p}+ \frac{1}{T}\sum_{|n|\geqslant R}|n|^{p} e^{-c_{4}|n|}+T^{7}\sum_{|n|\geqslant R} |n|^{p}e^{-2c_{3}\Psi(|n|)}\\
        &\lesssim_{\phi,K,C_{1},c_{1}} R^{p}+T^{7}\sum_{|n|\geqslant R} |n|^{p}e^{-2c_{3}\Psi(|n|)}.
    \end{split}
\end{equation*}
Choose $R$ such  that $\Psi(R)=\frac{3.9}{c_{3}} \log T$, that is $R=\Gamma(\frac{3.9}{c_{3}} \log T)$. Hence for sufficiently large $T\geqslant T_{0}(\phi,K,C_{1},C_{2},c_{1},c_{3},\epsilon_{0},p)$,
\begin{equation*}
    \langle |X_{H}|_{\phi}^p\rangle (T) \leqslant  \bigg[ \Gamma\bigg(\frac{4}{c_{3}}\log T\bigg)\bigg]^{p}.
\end{equation*}
This finishes the proof. 
\end{proof}

In particular, we are interested in the following $\Psi(\cdot)$.
\begin{corollary}\label{ca1}
    Under the assumptions of Theorem \ref{General Transport Bound}, if $\Psi(\cdot)$ takes the form
    \begin{equation*}
        \Psi(N) = N^{\delta}
    \end{equation*}
    for some $\delta>0$, then for any $\varepsilon>0$, there exists $T_1=T_1(T_0,\delta,\varepsilon)>0$ such that for all $T\geqslant T_{1}$,
    \begin{equation*}
        \langle |X_{H}|^{p}_{\phi}\rangle (T) \leqslant (\log T)^{\frac{p}{\delta}+\varepsilon}.
    \end{equation*}
\end{corollary}

\begin{corollary}\label{ca2}
    Under the assumptions of Theorem \ref{General Transport Bound}, if $\Psi(\cdot)$ takes the form
    \begin{equation*}
        \Psi(N)= \exp \bigg( \delta(\log N)^{\sigma} \bigg)
    \end{equation*}
    for some $\delta,\sigma>0$, then for any $\varepsilon>0$, there exists $T_{2}=T_{2}(T_{0},\delta,\sigma,\varepsilon)>0$ such that for all $T\geqslant T_{2}$,
    \begin{equation*}
        \langle |X_{H}|^{p}_{\phi}\rangle (T) \leqslant \exp \bigg[p \bigg(\frac{1+\varepsilon}{\delta}\log\log T\bigg)^{1/\sigma} \bigg].
    \end{equation*}
\end{corollary}

\begin{corollary}\label{ca3}
    Under the assumptions of Theorem \ref{General Transport Bound}, if $\Psi(\cdot)$ takes the form
    \begin{equation*}
        \Psi(N)= (\log N)^{1/\delta}
    \end{equation*}
    for some $0<\delta<1$, then for any $\varepsilon>0$, there exists $T_{3}=T_{3}(T_{0},\delta,\varepsilon)>0$ such that for all $T\geqslant T_{3}$,
    \begin{equation*}
        \langle |X_{H}|^{p}_{\phi}\rangle (T) \leqslant \exp \bigg(p (\log T)^{\delta+\varepsilon}\bigg).
    \end{equation*}
\end{corollary}


\section{Discrepancy estimates}\label{DiscrepancyEstimates}

We first prove a few fundamental results regarding discrepancy estimates for the Kronecker sequence $\{n\alpha\}$ on $\mathbb{T}^{d}$.

\begin{theorem}\label{dks}
    Let $\Phi:  \mathbb{R}^{+} \rightarrow \mathbb{R}^{+}$ be such that $\Phi(t)/t$ is monotone increasing on $[\rho,\infty)$ for some $\rho\geqslant 2$. Let $\mu:=\max_{t\in[1,\rho]} \Phi(t)<\infty$.   Let $\alpha\in\mathbb{R}^{d}$ satisfy
    \begin{equation}\label{Gen_Diophantine_Cond}
        \|\langle n,\alpha\rangle\|_{\mathbb{T}} > \frac{1}{\Phi(|n|)} \quad \text{for any}\ n\in\mathbb{Z}^{d}\setminus\{0\}.
    \end{equation}
    Then for any $M\geqslant \rho$,
    \begin{equation*}
    D_N(\alpha) \lesssim_{\rho,\mu,d} \frac{1}{M} + \frac{1}{N}+\frac{\Phi(M)\log(\Phi(M))(\log M)^d}{MN}.
\end{equation*}
\end{theorem}
\begin{proof}
 By Theorem \ref{ETI}, for any $M \in \mathbb{N}$ we have
\begin{equation*}
    D_{N}(\alpha) \lesssim_{d} \frac{1}{M}+ \frac{1}{N}\sum_{|m| = 1}^{M}\frac{1}{r(m)}\bigg|\sum_{k= 1}^{N} e^{2\pi i k \langle m,\alpha\rangle} \bigg|,
\end{equation*}
where $r(m)=\prod_{i=1}^{d}\max \{1,|m_{i}|\}$ and  $|m|=\max_{1\leqslant  i\leqslant  d} |m_{i}|$. Since 
\begin{equation*}
    \bigg|\sum_{k=1}^{N} e^{2\pi i k x}\bigg| \lesssim \min \{N, \|x\|_{\mathbb{T}}^{-1}\},
\end{equation*}
we have
\begin{equation*}
    \frac{1}{N}\sum_{|m| = 1}^{\rho}\frac{1}{r(m)}\bigg|\sum_{k= 1}^{N} e^{2\pi i k \langle m,\alpha\rangle} \bigg| \lesssim \frac{1}{N} \sum_{|m| = 1}^{\rho}\frac{1}{r(m) \cdot \|\langle m,\alpha\rangle\|_{\mathbb{T}}}\lesssim_{\rho,\mu} \frac{1}{N},
\end{equation*}
and thus
\begin{equation}\label{ET:simp}
    D_N (\alpha) \lesssim_{\rho,\mu,d} \frac{1}{M}+ \frac{1}{N}+\frac{1}{N}\sum_{|m|=\rho}^{M}\frac{1}{r(m) \cdot \|\langle m,\alpha\rangle\|_{\mathbb{T}}}=: \frac{1}{M}+\frac{1}{N}+\frac{S_{0}}{N}.
\end{equation}

Let $r=(r_{1},\cdots, r_{d}) \in \mathbb{N}^d$ such that 
\begin{equation*}
    \frac{\log \rho}{\log 2}\leqslant  r_i \leqslant  \frac{\log M}{\log 2} \quad \text{for each}\ 1\leqslant  i\leqslant  d.
\end{equation*}
Fix any $r$ and denote
\begin{equation*}
    T_{r}:=\{m\in\mathbb{Z}^d:   2^{r_i-1}\leqslant  |m_i| \leqslant  2^{r_i} \ \text{for} \ 1\leqslant  i\leqslant  d \}.
\end{equation*}
Thus
\begin{equation}\label{S_0:bound1}
    \begin{split}
        S_0 &=\sum_{r} \sum_{m\in T_{r}} \frac{1}{r(m) \cdot \|\langle m,\alpha\rangle\|_{\mathbb{T}}}\\
        &\leqslant  \sum_{r} 2^{-\sum_{i=1}^{d} (r_i - 1)}\sum_{m \in T_{r}}\frac{1}{\|\langle m,\alpha\rangle\|_{\mathbb{T}}}.
    \end{split}
\end{equation}

Without loss of generality, we assume $r$ satisfies $r_1  = |r|$. By (\ref{Gen_Diophantine_Cond}) and monotonicity of $\Phi(\cdot)$, for any $m \in T_{r}$ we have the estimate 
\begin{equation*}
    \|\langle m, \alpha\rangle\|_{\mathbb{T}}\geqslant \frac{1}{\Phi(|m|)}\geqslant \frac{1}{\Phi(2^{r_1})}.
\end{equation*}

Define $\Delta = \Phi(2^{r_1})$. We need the following result:

\begin{lemma}\label{FixedPoints}
For all $l\leqslant  \lfloor \Delta \rfloor$, there are at most $2^{d+1}$ points $m$ in $T_{r}$ satisfying the inequality  
\begin{equation*}
l \Delta^{-1} \leqslant  \|\langle m, \alpha\rangle\|_{\mathbb{T}} \leqslant  (l+1)\Delta^{-1}.
\end{equation*}
\end{lemma}

\begin{proof}
    Otherwise if there exist $(2^{d+1}+1)$ points, then by the pigeonhole principle, either there will be $(2^{d}+1)$ points satisfying
\begin{equation*}
\|\langle m,\alpha\rangle\|_{\mathbb{T}}=\{ \langle m,\alpha \rangle\},
\end{equation*}
or there will be $(2^{d}+1)$ points satisfying
\begin{equation*}
	\|\langle m,\alpha\rangle\|_{\mathbb{T}}=1-\{ \langle m,\alpha \rangle\}.
\end{equation*}
In either case, there exists a hyperoctant of $\mathbb{Z}^{d}$ such that at least two points $m,m'$ among $(2^{d}+1)$ points are located in this hyperoctant, which means
\begin{equation*}
	m_{j} m'_{j}\geqslant 0 \quad \text{for all} \ j=1,\cdots,d.
\end{equation*}
It would follow that 
\begin{equation*}
    0<|m-m'|<2^{r_{1}-1},
\end{equation*}
and 
\begin{equation*}
    \begin{split}
        \|\langle m-m',\alpha\rangle\|_{\mathbb{T}}&\leqslant  \{\langle m-m',\alpha\rangle\} \\
        &=|\{\langle m,\alpha\rangle\}-\{\langle m',\alpha\rangle\}|\\
        &=|\|\langle m,\alpha\rangle\|_{\mathbb{T}}-\|\langle m',\alpha\rangle\|_{\mathbb{T}}|\\
        &\leqslant  \Delta^{-1}.
    \end{split}
\end{equation*}
Using (\ref{Gen_Diophantine_Cond}), we arrive at
\begin{equation*}
\|\langle m-m',\alpha\rangle\|_{\mathbb{T}}\geqslant \frac{1}{\Phi(|m - m'| )}  >\frac{1}{\Phi(2^{r_1})}=\Delta^{-1},
\end{equation*}
which is a contradiction. 
\end{proof}

From Lemma \ref{FixedPoints}, we see that
\begin{equation*}
    \sum_{m\in T_{r}} \frac{1}{\|\langle m,\alpha\rangle\|_{\mathbb{T}}} \lesssim_{d}\Delta\sum_{l = 1}^{[\Delta]} \frac{1}{l} \lesssim_{d} \Delta \log(\Delta).
\end{equation*}
Substituting the above estimate into \eqref{S_0:bound1} then yields
\begin{equation*}
    S_{0}\lesssim_{d} \sum_{r_1=\log_{2}{\rho}}^{\log_{2} M}2^{-r_1}\Delta\log (\Delta) (\log M)^{d-1}.
\end{equation*}
By the assumption that $\Phi(t)/t$ is monotone increasing, we thus have
\begin{equation}\label{Me}
    S_0 \lesssim_{d} M^{-1}\Phi(M)\log(\Phi(M))(\log M)^d.
\end{equation}
Substituting \eqref{Me} into \eqref{ET:simp} gives
\begin{equation*}
    D_N(\alpha) \lesssim_{\rho,\mu,d} \frac{1}{M} + \frac{1}{N} +\frac{\Phi(M)\log(\Phi(M))(\log M)^d}{MN}.
\end{equation*}
This finishes the proof.
\end{proof}

In particular, we are interested in the following cases.
\begin{corollary}\label{cor_NlogN}
    Let $\eta>0, \gamma\geqslant 1$. Let $\alpha\in \mathbb{R}^{d}$ satisfy
    \begin{equation*}
        \|\langle n,\alpha\rangle\|_{\mathbb{T}} \geqslant \eta |n|^{-\gamma}\quad \text{for all}\ n\in\mathbb{Z}^{d}\setminus\{0\}.
    \end{equation*}
    Then for sufficiently large $N$,
    \begin{equation*}
        D_{N}(\alpha)\lesssim_{d,\gamma,\eta} N^{-1/\gamma} (\log N)^{d+1}.
    \end{equation*}
\end{corollary}
\begin{proof}
    Let $\Phi(t)= \eta^{-1} t^{\gamma}$. Then $\Phi(t)/t$ is monotone increasing on $[2,\infty)$. Apply Theorem \ref{dks} with $\Phi(M)=N$, that is,
    \begin{equation*}
        M=(\eta N)^{1/\gamma},
    \end{equation*}
    we have
    \begin{equation*}
        \begin{split}
            D_{N}(\alpha)&\lesssim_{d,\gamma,\eta}\frac{1}{M}+\frac{1}{N}+\frac{\log N (\log M)^{d}}{M}\\
            &\lesssim_{d,\gamma,\eta} N^{-1/\gamma} (\log N)^{d+1}.
        \end{split}
    \end{equation*}
\end{proof}

\begin{corollary}\label{cor_exp(log(n))}
    Let $\eta>0, \kappa>0,\gamma > 1$. Let $\alpha \in \mathbb{R}^d$ satisfy
    \begin{equation*}
        \|\langle n,\alpha\rangle\|_{\mathbb{T}} \geqslant \eta e^{-\kappa (\log |n|)^{\gamma}} \quad \text{for all}\ n\in\mathbb{Z}^{d}\setminus\{0\}.
    \end{equation*}
    Then for sufficiently large $N$,
    \begin{equation*}
        D_{N}(\alpha)\lesssim_{d,\kappa,\gamma,\eta} (\log N)^{1+d/\gamma} \exp \bigg[- \bigg(\frac{1}{\kappa}\log N \bigg)^{1/\gamma}\bigg].
    \end{equation*}    
\end{corollary}
\begin{proof}
    Let $\Phi(t)=\eta^{-1}e^{\kappa(\log t)^{\gamma}}$. Then there exists $\rho$ (depending on $\kappa,\gamma$) such that $\Phi(t)/t$ is monotone increasing on $[\rho,\infty)$. Apply Theorem \ref{dks} with $\Phi(M)=N$, that is, 
    \begin{equation*}
        M=  \exp \bigg[\bigg(\frac{1}{\kappa}\log (\eta N)\bigg)^{1/\gamma}\bigg],
    \end{equation*}
    we have
    \begin{equation*}
        \begin{split}
            D_{N}(\alpha)&\lesssim_{\kappa,\gamma,d,\eta} \frac{1}{M}+\frac{1}{N}+ \frac{\log N (\log M)^{d}}{M}\\
            &\lesssim_{d,\kappa,\gamma,\eta} (\log N)^{1+d/\gamma}\exp \bigg[- \bigg(\frac{1}{\kappa}\log N\bigg)^{1/\gamma}\bigg].
        \end{split}
    \end{equation*}
\end{proof}

\begin{corollary}\label{cor_logN}
    Let $\eta>0,\kappa>0,0<\gamma<1$. Let $\alpha\in\mathbb{R}^{d}$ satisfy
    \begin{equation*}
        \|\langle n,\alpha\rangle\|_{\mathbb{T}}\geqslant \eta e^{-\kappa|n|^{\gamma}}\quad \text{for all}\ n\in\mathbb{Z}^{d}\setminus\{0\}.
    \end{equation*}
    Then for sufficiently large $N$,
    \begin{equation*}
        D_{N}(\alpha)\lesssim_{d,\kappa,\gamma,\eta} (\log N)^{-1/\gamma}.
    \end{equation*}
\end{corollary}
\begin{proof}
    Let $\Phi(t)=\eta^{-1} e^{\kappa t^{\gamma}}$. Then there exists $\rho$ (depending on $\kappa,\gamma$) such that $\Phi(t)/t$ is monotone increasing on $[\rho,\infty)$. Apply Theorem \ref{dks} with $\Phi(M)=\sqrt{N}$, that is,
    \begin{equation*}
        M=\bigg(\frac{1}{\kappa} \log (\eta\sqrt{N}) \bigg)^{1/\gamma},
    \end{equation*}
    we have
    \begin{equation*}
        \begin{split}
            D_{N}(\alpha)&\lesssim_{d,\kappa,\gamma,\eta}\frac{1}{M}+\frac{1}{N}+\frac{\log N (\log M)^{d}}{M\sqrt{N}}\\
            &\lesssim_{d,\kappa,\gamma,\eta} (\log N)^{-1/\gamma}.
        \end{split}
    \end{equation*}
\end{proof}

The following semi-algebraic discrepancy estimate is useful and will be applied repeatedly.
\begin{theorem}\label{semidks}
    Suppose $D_{N}(x_{n})\leqslant  Y_{N}\xrightarrow{N\rightarrow \infty} 0$. Let $\mathcal{S}\subseteq [0,1]^{d}$ be a semi-algebraic set with $\deg(\mathcal{S})\leqslant B$ and $\leb (\mathcal{S})\leqslant Y_{N}$. Then
    \begin{equation*}
        \#\{x_{n}\in \mathcal{S}\} \leqslant  2 B^{C(d)} N Y_{N}^{1/d}.
    \end{equation*}
\end{theorem}
\begin{proof}
    Let $\epsilon = Y_{N}^{1/d}$. Then
    \begin{equation*}
        \leb(\mathcal{S}) < Y_{N} = \epsilon^{d}.
    \end{equation*}
    By Lemma~\ref{covers}, the set $\mathcal{S}$ can be covered by at most $B^{C(d)} \epsilon^{1-d}$ balls of radius $\epsilon$. Let $D$ be one such ball. Then by the definition of discrepancy,
    \begin{equation*}
        \#\{x_{n} \in D\} \leqslant N \leb(D) + N D_{N}(\alpha) \leqslant 2N Y_{N}.
    \end{equation*}
    Summing over all such balls, we obtain
    \begin{equation*}
        \#\{x_{n} \in \mathcal{S}\} \leqslant 2 B^{C(d)} N Y_{N}^{1/d}.
    \end{equation*}
    This finishes the proof.
\end{proof}
Theorem \ref{semidks} was proved for $Y_{N}=N^{-\varsigma}$ with $\varsigma>0$ in \cite{MR4546503}.

\section{Large deviation theorem for Liouville frequencies}\label{LargeDeviationTheoremForLiouvilleFrequency}
Since we are interested in quantum dynamics with Liouville frequencies, we will need the following large deviation theorem for transfer matrices, which holds in the Liouville setting.

\begin{theorem}[\cite{MR4373222}]\label{HZ}
    Let $\alpha\in \mathbb{R}\setminus \mathbb{Q}$ and $V\in C^{\omega}_{h}(\mathbb{T},\mathbb{R})$. There exist constants $\tilde{c}_{1}(V,h), \tilde{c}_{2}(V,h)\in (0,1)$ such that, if \footnote{Recall that $\beta(\alpha)=\limsup_{|k|\rightarrow \infty} -\frac{\log \|k\alpha\|_{\mathbb{T}}}{|k|}$.}
    \begin{equation*}
        0\leqslant  \beta(\alpha)< \tilde{c}_{1}\inf_{E\in [a,b]} L(E),
    \end{equation*}
    then there exists $N_{1}=N_{1}(\alpha, \inf_{E\in [a,b]} L(E), V, h)>0$ such that for any $N\geqslant N_{1}$, the following large deviation estimates hold uniformly in $E\in [a,b]$:
    \begin{enumerate}
        \item If $0<L(E)<1$, then
        \begin{equation*}
            \leb \bigg\{\theta: \bigg|\frac{1}{N}\log \|M_{N}(\theta)\|- L_{N}(E)\bigg|> \frac{1}{100} L(E)\bigg\}< e^{-\tilde{c}_{2} L(E)N}.
        \end{equation*}

        \item If $L(E)\geqslant 1$, then
        \begin{equation*}
            \leb \bigg\{\theta: \bigg|\frac{1}{N}\log \|M_{N}(\theta)\|- L_{N}(E)\bigg|> \frac{1}{100} L(E)\bigg\}< e^{-\tilde{c}_{2} L(E)^{2}N}.
        \end{equation*}
    \end{enumerate}
\end{theorem}

\begin{remark}
    In fact, $\frac{1}{100}$ in Theorem \ref{HZ} can be replaced by any $0<\kappa<1$. See Remark 1.2 in \cite{MR4373222}.
\end{remark}

The conclusion of Theorem~\ref{HZ} remains valid for complex energies $E + i\epsilon$, where $E \in [a,b]$ and $|\epsilon| \leqslant \epsilon_{0}$ for some sufficiently small $\epsilon_{0} > 0$ because the proof of Theorem~\ref{HZ}  only involves the subharmonicity of the Lyapunov exponent. More precisely, we have the following:

\begin{corollary}\label{HZcoro}
    Under the same assumptions as in Theorem \ref{HZ}, there exist  $N_{1}=N_{1}(\alpha,V,a,b)>0$ and $\epsilon_{0}=\epsilon_{0}(\alpha,V,a,b)>0$ such that for any $N\geqslant N_{1}$, the following holds uniformly for $z=E+i\epsilon$ where $E\in[a,b]$ and $|\epsilon|\leqslant \epsilon_{0}$:
    \begin{equation*}
        \leb \bigg\{\theta: \bigg|\frac{1}{N}\log \|M_{N}(\theta)\|- L_{N}(z)\bigg|> \frac{1}{100} L(z)\bigg\}< e^{-\tilde{c}_{2} L(z)N}.
    \end{equation*}
\end{corollary}

With Theorem \ref{HZ} and Corollary \ref{HZcoro} in hand, we deduce the following large deviation theorem for Green's function in the Liouville setting.
\begin{theorem}\label{LDTgreen}
    Let $\alpha\in\mathbb{R}\setminus\mathbb{Q}$ and $V\in C^{\omega}_{h}(\mathbb{T},\mathbb{R})$. There exist constants $\tilde{c}_{1}(V,h), \tilde{c}_{2}(V,h)\in (0,1)$ such that, if
    \begin{equation*}
        0\leqslant \beta(\alpha)<\tilde{c}_{1} \inf_{E\in [a,b]} L(E),
    \end{equation*}
    there exist $N_{2}=N_{2}(\alpha, V, a,b)>0$ and $\epsilon_{0}=\epsilon_{0} (\alpha, V, a,b)>0$ such that for any $N\geqslant N_{2}$, and $z=E+i\epsilon$ with $E\in[a,b]$ and $|\epsilon|\leqslant \epsilon_{0}$, the following holds. There exists a subset $\Theta_{N}\subseteq \mathbb{T}$ (depending on $z$) with
    \begin{equation*}
        \leb (\Theta_{N})\leqslant  e^{-\tilde{c}_{2} L(z) (2N+1)},\quad \deg (\Theta_{N})\lesssim_{h} L(z)^{2} N^{5},
    \end{equation*}
    such that for any $\theta\notin \Theta_{N}$, one of the intervals
    \begin{equation*}
        \Lambda=[-N,N]; [-N,N-1]; [-N+1,N]; [-N+1,N-1]
    \end{equation*}
    will satisfy
    \begin{equation*}
        |G_{\Lambda}(z)(m,n)|\leqslant e^{-\frac{1}{2}L(z)|m-n|}, \quad \text{for any}\ |m-n|>\frac{|\Lambda|}{20}.
    \end{equation*}
\end{theorem}
\begin{proof}
    By Corollary \ref{HZcoro}, for any $N\geqslant N_{1}(\alpha,V,a,b)$, there exists a subset $\Theta_{N}\subseteq \mathbb{T}$ with $\leb(\Theta_{N})\leqslant e^{-\tilde{c}_{2}L(z) (2N+1)}$ such that for any $\theta\notin \Theta_{N}$, 
    \begin{equation*}
        \|M_{[-N,N]}(\theta)\|\geqslant e^{\frac{99}{100}(2N+1)L(z)}.
    \end{equation*}
    Then it follows from \eqref{four} that for one of the intervals
    \begin{equation*}
        \Lambda=[-N,N]; [-N,N-1]; [-N+1,N]; [-N+1,N-1]
    \end{equation*} 
    we have for any $\theta\notin\Theta_{N}$,
    \begin{equation}\label{db}
        |\det(H_{\Lambda}(\theta)-z)|\geqslant \frac{1}{4}e^{\frac{99}{100}(2N+1)L(z)}.
    \end{equation}
    By \eqref{fur} and the compactness argument, for any $\varepsilon>0$, there exists $\bar{N}_{1}=\bar{N}_{1}(\alpha,V,a,b,\varepsilon)$ such that for any $N\geqslant \bar{N}_{1}$, we have
    \begin{equation}\label{nb}
        \sup_{\theta\in\mathbb{T}}\|M_{N}(\theta)\|\leqslant e^{N(L(z)+\varepsilon)}.
    \end{equation}
    It follows from \eqref{cramer},\eqref{db} and \eqref{nb} that for any $\theta\notin\Theta_{N}$, we have 
    \begin{equation*}
        |G_{\Lambda}(\theta)(m,n)|\leqslant e^{\frac{1}{100}|\Lambda|L(z)} e^{|\Lambda|\varepsilon} e^{-L(z)|m-n|},\ \text{for any}\ m,n\in\Lambda.
    \end{equation*}
    In particular, for any $\theta\notin\Theta_{N}$ and $m,n\in\Lambda$ with $|m-n|\geqslant |\Lambda|/20$, 
    \begin{equation*}
        |G_{\Lambda}(\theta)(m,n)|\leqslant  e^{-\frac{1}{2}L(z)|m-n|}.
    \end{equation*}

    In the following, we estimate the complexity. Consider the property 
    \begin{equation}\label{P}
        |G_{\Lambda}(\theta)(m,n)|\leqslant e^{\frac{1}{100}|\Lambda|L(z)} e^{|\Lambda|\varepsilon} e^{-L(z)|m-n|},\ \text{for any}\ m,n\in\Lambda.
    \end{equation}
    Since $V\in C^{\omega}_{h}(\mathbb{T},\mathbb{R})$, we write $V(\theta)=\sum_{k\in\mathbb{Z}} \widehat{V}_{k} e^{2\pi i k \theta}$. Consider its Fourier truncation $V_{1}(\theta)=\sum_{|k|\leqslant C' N} \widehat{V}_{k}e^{2\pi i k \theta}$ where $C'= \frac{100}{2\pi h}L(z)$. Then
    \begin{equation*}
        \sup_{\theta\in\mathbb{T}} |V_{1}(\theta)-V(\theta)|\leqslant \|V\|_{h}e^{-2\pi hC'N} \leqslant \|V\|_{h} e^{-100L(z)N}.
    \end{equation*}
    By \eqref{four} and a telescoping argument, for any $1\leqslant J\leqslant |\Lambda|$, we have
    \begin{equation*}
        \begin{split}
            \sup_{\theta\in\mathbb{T}} & |\det (H_{J}(\theta,V)-z)-\det (H_{J}(\theta,V_{1})-z)|\\
            &\leqslant \sup_{\theta\in\mathbb{T}} \bigg\| \prod_{j=J}^{1} S_{z}^{V}(\theta+j\alpha) - \prod_{j=J}^{1} S_{z}^{V_{1}}(\theta+j\alpha)\bigg\|\\
            &\leqslant  \sup_{\theta\in\mathbb{T}} |V_{1}(\theta)-V(\theta)| Je^{J(L(z)+\varepsilon)} \\
            &\leqslant e^{-20 |\Lambda| L(z)},
        \end{split}
    \end{equation*}
    which implies for any $\theta\notin \Theta_{N}$,
    \begin{equation*}
        |G_{\Lambda}(\theta, V)(m,n)-G_{\Lambda}(\theta, V_{1})(m,n)| \leqslant e^{-10|\Lambda|L(z)}.
    \end{equation*}
    Thus we may substitute $V$ by $V_{1}$ in \eqref{P}, that is,
    \begin{equation}\label{PP}
        e^{L(z)|m-n|}|\det (H_{\Lambda}(V_{1},\theta)-E)_{m,n} | \leqslant \frac{1}{2} e^{\frac{1}{100}|\Lambda|L(z)} e^{|\Lambda|\varepsilon} |\det (H_{\Lambda}(V_{1},\theta)-z)|,
    \end{equation}
    where $M_{m,n}$ denotes the $(m,n)$-minor of matrix $M$. Using the Hilbert-Schmidt norm, we may consider the tighter inequality,
    \begin{equation}\label{PPP}
        \begin{split}
            \sum_{m,n\in\Lambda}  e^{2L(z)|m-n|} &\Big(\det (H_{\Lambda}(V_{1},\theta)-z)_{m,n} \Big)^{2} \\
            &\leqslant \frac{1}{4} e^{\frac{1}{50}|\Lambda|L(z)} e^{|\Lambda|\varepsilon}\Big(\det (H_{\Lambda}(V_{1},\theta)-z)\Big)^{2}.
        \end{split}
    \end{equation}
    Clearly, \eqref{PPP} is of the form
    \begin{equation}\label{P1}
        P_{1}(\cos 2\pi \theta, \sin 2\pi \theta) \geqslant 0,
    \end{equation}
    where $P_{1}$ is a polynomial with degree at most $C'^{2}N^{4}$. One further truncates ``$\cos$'' and ``$\sin$'' by a Taylor polynomial of degree at most $N$, and get a tighter inequality of the form
    \begin{equation}\label{P2}
        P_{2}(\theta)\geqslant 0,
    \end{equation}
    where the degree of $P_{2}$ is at most $C'^{2}N^{5}$. It is obvious that \eqref{P2} implies \eqref{P1}, which further implies \eqref{P}.
    
    Denote ${\Theta}'_{N}=\{\theta\in [0,1]: P_{2}(\theta)< 0\}$. By the definition of the degree of semi-algebraic sets (see Definition \ref{defsemi}), we have $\deg(\overline{\Theta'_{N}})\leqslant C'^{2}N^{5}$.

    For any $\theta\notin\Theta_{N}$, it is easy to check that $P_{2}(\theta)\geqslant 0$ (with $\varepsilon$ replaced as $2\varepsilon$). Thus ${\Theta}'_{N}\subseteq \Theta_{N}$ and
    \begin{equation*}
        \leb(\overline{\Theta'_{N}}) \leqslant \leb(\Theta_{N})\leqslant e^{-\tilde{c}_{2}L(z)(2N+1)}. 
    \end{equation*}
    Finally one may replace $\Theta_{N}$ by $\overline{\Theta'_{N}}$. This finishes the proof.
\end{proof}

\section{Proof of main results}\label{Proof of main Results-Section}

We are now ready to prove our main results by combining the tools developed in Sections \ref{CriterionForQuantumDynamics}, \ref{DiscrepancyEstimates} and \ref{LargeDeviationTheoremForLiouvilleFrequency}.
\subsection{Proof of Theorem \ref{qdDC}}
\begin{proof}
    Fix a small $\varepsilon > 0$ and take $N$ sufficiently large. Let $C\geqslant 1$ be the constant $C(d)$ from Lemma \ref{covers} when $d=1$. Define an interval $I\subseteq \mathbb{Z}$ centered at $0$ and of length $|I| = N^{\delta}$ with $\delta = \frac{1}{5\gamma C} - \varepsilon$. Set $I_j = I + j$ for $|j| \leqslant N$. 

    Recall that for any $\theta \in \mathbb{T}$, one has $H_{I_j}(\theta) = H_{I}(\theta + j\alpha)$. Then for any $z = E + \frac{i}{T}$ with sufficiently large $T$,
    \begin{equation}\label{shift}
        G_{I_j}(z, \theta) = G_{I}(z, \theta + j\alpha).
    \end{equation}
    By Theorem \ref{LDTgreen}, there exists a subset $\Theta_{I} \subseteq \mathbb{T}$ such that $\leb(\Theta_{I}) \leqslant e^{-\tilde{c}_2 L(z)|I|}$ and $\deg(\Theta_{I}) \lesssim_{h} L(z)^2 |I|^5$, with the property that for any $\theta \notin \Theta_{I}$ and $|m-n| \geqslant |I|/20$,
    \begin{equation*}
        |G_{I}(z, \theta)(m,n)| \leqslant e^{-\frac{1}{2} L(z) |m-n|} \leqslant e^{-\frac{ \lambda}{40} |I|}.
    \end{equation*}

    Applying Theorem \ref{semidks} and Corollary \ref{cor_NlogN}, for any $\theta \in \mathbb{T}$, we obtain
    \begin{equation*}
        \begin{split}
            \# \big\{ |j| \leqslant N \,\colon\, \theta + j\alpha \bmod \mathbb{Z} \in \Theta_{I} \big\} 
            &\lesssim_{\gamma, \eta} (\deg \Theta_{I})^C N N^{-\frac{1}{\gamma}} (\log N)^2 \\
            &\lesssim_{\gamma, \eta, h} L(z)^{2C} N^{1 + 5\delta C - \frac{1}{\gamma}} (\log N)^2 \\
            &= o(N),
        \end{split}
    \end{equation*}
    which implies that there exists some $|j'| \leqslant N$ such that $\theta + j'\alpha \bmod \mathbb{Z} \notin \Theta_{I}$ for sufficiently large $N$. Consequently,
    \begin{equation*}
        |G_{I}(z, \theta + j'\alpha)(m,n)| \leqslant e^{-\frac{\lambda}{40} |I|}.
    \end{equation*}
    Moreover, by \eqref{shift} and $|I_{j'}| = |I|$, it follows that
    \begin{equation*}
        |G_{I_{j'}}(z, \theta)(m,n)| \leqslant e^{-\frac{ \lambda}{40} |I_{j'}|}.
    \end{equation*}

    Finally, applying Corollary \ref{ca1} with $\Psi(N) = N^{\delta}$ where $\delta = \frac{1}{5\gamma C} - \varepsilon$, we conclude
    \begin{equation*}
        \sup_{\theta\in\mathbb{T}}\langle |X_{H_{\theta}}|^p_{\phi} \rangle (T) \leqslant (\log T)^{5p\gamma C + \varepsilon}.
    \end{equation*}
    This completes the proof.
\end{proof}

\subsection{Proof of Theorem \ref{qdWDC}}
\begin{proof}
   We present only the essential steps, as the proof is similar to that of Theorem \ref{qdDC}.

   Fix a small $\varepsilon > 0$ and take $N$ sufficiently large. Let $C \geqslant 1$ be the constant $C(d)$ from Lemma \ref{covers} when $d=1$. Define an interval $I$ centered at $0$ and of length 
   \begin{equation*}
       |I| = \exp\big[ \delta (\log N)^{1/\gamma} \big]
   \end{equation*}
   with 
   \begin{equation*}
       \delta = \frac{1}{5C} \bigg[ \bigg(\frac{1}{\kappa} \bigg)^{1/\gamma} - \varepsilon \bigg].
   \end{equation*}

   Applying Theorem \ref{semidks} together with Corollary \ref{cor_exp(log(n))}, we obtain that for any $\theta \in \mathbb{T}$,
   \begin{equation*}
       \begin{split}
           \# \big\{ |j| \leqslant N & : \theta + j\alpha \bmod \mathbb{Z} \in \Theta_{I} \big\} \\
           &\lesssim_{\gamma,\kappa,\eta} (\deg \Theta_{I})^{C} N (\log N)^{1+\frac{1}{\gamma}} 
           \exp \bigg[ - \bigg( \frac{1}{\kappa} \log N \bigg)^{1/\gamma} \bigg] \\
           &\lesssim_{\gamma,\kappa,\eta,h} L(z)^{2C} N 
           \exp \big[ 2\log\log N - \varepsilon (\log N)^{1/\gamma} \big] \\
           &= o(N),
       \end{split}
   \end{equation*}
   which implies that there exists $|j'| \leqslant N$ such that
   \begin{equation*}
       |G_{I_{j'}}(z, \theta)(m, n)| \leqslant e^{ -\frac{ \lambda}{40} |I_{j'}| }.
   \end{equation*}

   Applying Corollary \ref{ca2} with $\Psi(N) = \exp \big[ \delta (\log N)^{1/\gamma} \big]$, we conclude
   \begin{equation*}
       \begin{split}
           \sup_{\theta \in \mathbb{T}} \langle |X_{H_{\theta}}|^p_{\phi} \rangle (T) 
       &\leqslant \exp \left[ p \left( \frac{5C + \varepsilon}{(1/\kappa)^{1/\gamma}} \log\log T \right)^{\gamma} \right] \\
       &\leqslant \exp \bigg( p\kappa (5C+\varepsilon)^{\gamma} (\log\log T)^{\gamma}\bigg).
       \end{split}
   \end{equation*}
\end{proof}

\subsection{Proof of Theorem \ref{qdLiou}}
\begin{proof}
    We present only the essential steps, as the proof is similar to that of Theorem \ref{qdDC}.

    Denote by $C \geqslant 1$ the constant $C(d)$ from Lemma \ref{covers} when $d=1$. Since $\gamma <  \frac{1}{5C}$, fix a small $\varepsilon > 0$ such that $5C\gamma + \varepsilon < 1$. Let $N$ be sufficiently large, and define an interval $I$ centered at $0$ with length
\begin{equation*}
    |I| = (\log N)^{1/\delta},
\end{equation*}
where
\begin{equation*}
    \delta = 5C\gamma + \varepsilon.
\end{equation*}

Applying Theorem \ref{semidks} in combination with Corollary \ref{cor_logN}, we obtain that for any $\theta \in \mathbb{T}$,
\begin{equation*}
    \begin{split}
        \# \{ |j| \leqslant N : \theta + j\alpha \bmod \mathbb{Z} \in \Theta_{I} \}
        &\lesssim_{\kappa,\gamma,\eta} (\deg \Theta_{I})^C N (\log N)^{-1/\gamma} \\
        &\lesssim_{\kappa,\gamma,\eta,h} L(z)^{2C} N (\log N)^{\frac{5C}{\delta} - \frac{1}{\gamma}} \\
        &= o(N),
    \end{split}
\end{equation*}
which implies that there exists $|j'| \leqslant N$ such that
\begin{equation*}
    |G_{I_{j'}}(z, \theta)(m, n)| \leqslant e^{ -\frac{\lambda}{40} |I_{j'}| }.
\end{equation*}

Applying Corollary \ref{ca3} with $\Psi(N) = (\log N)^{1/\delta}$ yields
\begin{equation*}
    \sup_{\theta \in \mathbb{T}} \langle |X_{H_{\theta}}|^p_{\phi} \rangle (T) 
    \leqslant \exp \bigg( p (\log T)^{5C\gamma + \varepsilon} \bigg).
\end{equation*}
\end{proof}

\section*{Acknowledgments}
This work was carried out at the Texas A\&M Mathematics REU 2025 under the thematic program ``Semi-Algebraic Geometry in Schrodinger Equations", supported by the NSF REU  Site  grant DMS-2150094 and the Department of Mathematics at Texas A\&M University. We would like to thank Anne Shiu and Maurice Rojas for organizing and overseeing the REU program.

W.L. was supported in part by NSF grants DMS-2246031 and DMS-2052572, by a Simons Fellowship in Mathematics, and by a Visiting Miller Professorship from the Miller Institute for Basic Research in Science, University of California, Berkeley. W.L. thanks the Department of Mathematics and the Miller Institute for Basic Research in Science at UC Berkeley for their hospitality during his visit in Fall 2025, when part of this work was completed.




\section*{Statements and Declarations}
{\bf Conflict of Interest} 
The authors declare no conflicts of interest.

\vspace{0.2in}
{\bf Data Availability}
Data sharing is not applicable to this article as no new data were created or analyzed in this study.

\bibliographystyle{alpha}
\bibliography{main}

\newcommand{\etalchar}[1]{$^{#1}$}
\begin{thebibliography}{LPT{\etalchar{+}}25}

\bibitem[Aiz94]{MR1301371}
M.~Aizenman.
\newblock Localization at weak disorder: {S}ome elementary bounds.
\newblock {\em Rev. Math. Phys.}, 6(5A):1163--1182, 1994.
\newblock Special issue dedicated to Elliott H. Lieb.

\bibitem[BJ00]{MR1796713}
J.~Bourgain and S.~Jitomirskaya.
\newblock Anderson localization for the band model.
\newblock In {\em Geometric aspects of functional analysis}, volume 1745 of {\em Lecture Notes in Math.}, pages 67--79. Springer, Berlin, 2000.

\bibitem[BN19]{MR3990603}
G.~Binyamini and D.~Novikov.
\newblock Complex cellular structures.
\newblock {\em Ann. of Math. (2)}, 190(1):145--248, 2019.

\bibitem[Bou05]{MR2100420}
J.~Bourgain.
\newblock {\em Green's function estimates for lattice {S}chr\"odinger operators and applications}, volume 158 of {\em Annals of Mathematics Studies}.
\newblock Princeton University Press, Princeton, NJ, 2005.

\bibitem[DF22]{damanik2022one}
D.~Damanik and J.~Fillman.
\newblock {\em One-Dimensional Ergodic Schr{\"o}dinger Operators: I. General Theory}, volume 221.
\newblock American Mathematical Society, 2022.

\bibitem[DT97]{MR1470456}
M.~Drmota and R.F. Tichy.
\newblock {\em Sequences, discrepancies and applications}, volume 1651 of {\em Lecture Notes in Mathematics}.
\newblock Springer-Verlag, Berlin, 1997.

\bibitem[DT07]{MR2291919}
D.~Damanik and S.~Tcheremchantsev.
\newblock Upper bounds in quantum dynamics.
\newblock {\em J. Amer. Math. Soc.}, 20(3):799--827, 2007.

\bibitem[DT08]{MR2464194}
D.~Damanik and S.~Tcheremchantsev.
\newblock Quantum dynamics via complex analysis methods: general upper bounds without time-averaging and tight lower bounds for the strongly coupled {F}ibonacci {H}amiltonian.
\newblock {\em J. Funct. Anal.}, 255(10):2872--2887, 2008.

\bibitem[Gro87]{MR880035}
M.~Gromov.
\newblock Entropy, homology and semialgebraic geometry.
\newblock {\em Ast\'erisque}, (145-146):5, 225--240, 1987.
\newblock S\'eminaire Bourbaki, Vol.\ 1985/86.

\bibitem[GSB99]{MR1663518}
I.~Guarneri and H.~Schulz-Baldes.
\newblock Lower bounds on wave packet propagation by packing dimensions of spectral measures.
\newblock {\em Math. Phys. Electron. J.}, 5:Paper 1, 16, 1999.

\bibitem[GYZ23]{MR4637128}
L.~Ge, J.~You, and Q.~Zhou.
\newblock Exponential dynamical localization: criterion and applications.
\newblock {\em Ann. Sci. \'Ec. Norm. Sup\'er. (4)}, 56(1):91--126, 2023.

\bibitem[HZ22]{MR4373222}
R.~Han and S.~Zhang.
\newblock Large deviation estimates and {H}\"older regularity of the {L}yapunov exponents for quasi-periodic {S}chr\"odinger cocycles.
\newblock {\em Int. Math. Res. Not. IMRN}, (3):1666--1713, 2022.

\bibitem[JKL20]{MR4216568}
S.~Jitomirskaya, H.~Kr\"uger, and W.~Liu.
\newblock Exact dynamical decay rate for the almost {M}athieu operator.
\newblock {\em Math. Res. Lett.}, 27(3):789--808, 2020.

\bibitem[JL21]{MR4288185}
S.~Jitomirskaya and W.~Liu.
\newblock Upper bounds on transport exponents for long-range operators.
\newblock {\em J. Math. Phys.}, 62(7):Paper No. 073506, 9, 2021.

\bibitem[JL24]{MR4756946}
S.~Jitomirskaya and W.~Liu.
\newblock Universal reflective-hierarchical structure of quasiperiodic eigenfunctions and sharp spectral transition in phase.
\newblock {\em J. Eur. Math. Soc. (JEMS)}, 26(8):2797--2836, 2024.

\bibitem[JLM24]{jitomirskaya2024sharp}
S.~Jitomirskaya, W.~Liu, and L.~Mi.
\newblock Sharp palindromic criterion for semi-uniform dynamical localization.
\newblock {\em preprint, arXiv:2410.21700}, 2024.

\bibitem[JP22]{JPo}
S.~Jitomirskaya and M.~Powell.
\newblock Logarithmic quantum dynamical bounds for arithmetically defined ergodic {S}chr{\"o}dinger operators with smooth potentials.
\newblock In {\em Analysis at Large: Dedicated to the Life and Work of Jean Bourgain}, pages 173--201. Springer, 2022.

\bibitem[JS94]{JitoSimonSC}
S.~Jitomirskaya and B.~Simon.
\newblock Operators with singular continuous spectrum. {III}. {A}lmost periodic {S}chr\"odinger operators.
\newblock {\em Comm. Math. Phys.}, 165(1):201--205, 1994.

\bibitem[JZ22]{MR4404788}
S.~Jitomirskaya and S.~Zhang.
\newblock Quantitative continuity of singular continuous spectral measures and arithmetic criteria for quasiperiodic {S}chr\"odinger operators.
\newblock {\em J. Eur. Math. Soc. (JEMS)}, 24(5):1723--1767, 2022.

\bibitem[Las96]{MR1423040}
Y.~Last.
\newblock Quantum dynamics and decompositions of singular continuous spectra.
\newblock {\em J. Funct. Anal.}, 142(2):406--445, 1996.

\bibitem[Liu22]{MR4546503}
W.~Liu.
\newblock Quantitative inductive estimates for {G}reen's functions of non-self-adjoint matrices.
\newblock {\em Anal. PDE}, 15(8):2061--2108, 2022.

\bibitem[Liu23]{MR4564259}
W.~Liu.
\newblock Power law logarithmic bounds of moments for long range operators in arbitrary dimension.
\newblock {\em J. Math. Phys.}, 64(3):Paper No. 033508, 11, 2023.

\bibitem[LP22]{MR4578350}
M.~Landrigan and M.~Powell.
\newblock Fine dimensional properties of spectral measures.
\newblock {\em J. Spectr. Theory}, 12(3):1255--1293, 2022.

\bibitem[LPT{\etalchar{+}}25]{prema}
W.~Liu, M.~Powell, Y.~Tang, X.~Wang, R.~Zhang, and J.~Zhou.
\newblock {S}emi-algebraic discrepancy estimates for multi-frequency shift sequences with applications to quantum dynamics.
\newblock {\em preprint, arXiv:2507.19783}, 2025.

\bibitem[LPW24]{LPW}
W.~Liu, M.~Powell, and X.~Wang.
\newblock Quantum dynamical bounds for long-range operators with skew-shift potentials.
\newblock {\em preprint, arXiv:2411.00176}, 2024.

\bibitem[SS23]{MR4604835}
M.~Shamis and S.~Sodin.
\newblock Upper bounds on quantum dynamics in arbitrary dimension.
\newblock {\em J. Funct. Anal.}, 285(7):Paper No. 110034, 20, 2023.

\bibitem[Yom87]{MR889980}
Y.~Yomdin.
\newblock {$C^k$}-resolution of semialgebraic mappings. {A}ddendum to: ``{V}olume growth and entropy''.
\newblock {\em Israel J. Math.}, 57(3):301--317, 1987.

\end{thebibliography}
\end{document}